\documentclass[letterpaper, 10 pt, conference]{ieeeconf}  % Comment this line out
\IEEEoverridecommandlockouts
\usepackage{balance}
\usepackage{color}
\usepackage{caption}
\usepackage{subcaption}
\usepackage{enumerate}
\usepackage{cite}
\usepackage{empheq}
\usepackage{mathrsfs}

\usepackage{mathtools}

\usepackage{tikz}
\usepackage{circuitikz}

\usepackage{flushend}
\makeatletter
\pgfcircdeclarebipole{}{\ctikzvalof{bipoles/interr/height 2}}{spst}{\ctikzvalof{bipoles/interr/height}}{\ctikzvalof{bipoles/interr/width}}{

    \pgfsetlinewidth{\pgfkeysvalueof{/tikz/circuitikz/bipoles/thickness}\pgfstartlinewidth}

    \pgfpathmoveto{\pgfpoint{\pgf@circ@res@left}{0pt}}
    \pgfpathlineto{\pgfpoint{.6\pgf@circ@res@right}{\pgf@circ@res@up}}
    \pgfusepath{draw}   
}

\def\pgf@circ@spst@path#1{\pgf@circ@bipole@path{spst}{#1}}
\tikzset{switch/.style = {\circuitikzbasekey, /tikz/to path=\pgf@circ@spst@path, l=#1}}
\tikzset{spst/.style = {switch = #1}}
\makeatother

\makeatletter
\let\proof\@undefined                        % undefine \proof
\let\endproof\@undefined                  % undefine \endproof
\makeatother
\usepackage{graphicx,amssymb,amstext,amsmath,amsthm}

\usepackage[bookmarks=true]{hyperref}
\usepackage{algorithm,algorithmicx,algpseudocode}
\algnewcommand{\algorithmicgoto}{\textbf{go to}}%
\algnewcommand{\Goto}[1]{\algorithmicgoto~\ref{#1}}%
\algnewcommand{\LineComment}[1]{\Statex \(\triangleright\) #1}
\algnewcommand{\LineCommentN}[1]{\Statex \hspace{1cm}\(\triangleright\) #1}

\usepackage{multirow}
\usepackage{stfloats}

\newcommand{\argmax}{\operatornamewithlimits{arg\ max}}

\newtheorem{prop}{Proposition} % this could go into the preamble

\newtheorem{thm}{Theorem}
	\newtheorem{assumption}{Assumption}
\newtheorem{lem}{Lemma}
\newtheorem{defn}{Definition}
\newtheorem{rem}{Remark}
\newtheorem{problem}{Problem}

\usepackage{setspace}
\let\oldbibliography\thebibliography
\renewcommand{\thebibliography}[1]{%
  \oldbibliography{#1}%
}

\newcommand{\moh}[1]{{\color{black} #1}}

\newcommand{\moha}[1]{{\color{black} #1}}
\begin{document}

% paper title
\title{\LARGE \bf Computationally Efficient State and Model Estimation \\ via Interval Observers for Partially Unknown Systems} % Matrices} }

% You will get a Paper-ID when submitting a pdf file to the conference system
%\author{Author Names Omitted for Anonymous Review. Paper-ID Sze Zheng Yong}
\author{%
Mohammad Khajenejad and Zeyuan Jin \\
\thanks{%$^1$ The authors are with the Laboratory for Information and Decision Systems,
%Massachusetts Institute of Technology, Cambridge, MA, USA (e-mail: szyong@mit.edu, mhzhu@mit.edu, frazzoli@mit.edu).
 M. Khajenejad is with the Department of Mechanical Engineering, The University of Tulsa, Tulsa, OK, USA (e-mail: mok7673@utulsa.edu). Z. Jin is with SafeAI, Inc., Santa Clara, CA, USA (e-mail: zjin43@asu.edu).}
%\thanks{This work is partially supported by NSF grants CNS-1932066 and CNS-1943545.}
%\vspace{-0.35cm}
}

\maketitle
\thispagestyle{empty}
\pagestyle{empty}

\begin{abstract}
This paper addresses the synthesis of interval observers for partially unknown nonlinear systems subject to bounded noise, aiming to simultaneously estimate system states and learn a model of the unknown dynamics. Our approach leverages Jacobian sign-stable (JSS) decompositions, tight decomposition functions for nonlinear systems, and a data-driven over-approximation framework to construct interval estimates that provably enclose the true augmented states. By recursively computing tight and tractable bounds for the unknown dynamics based on current and past interval framers, we systematically integrate these bounds into the observer design. Additionally, we formulate semi-definite programs (SDP) for observer gain synthesis, ensuring input-to-state stability and optimality of the proposed framework. Finally, simulation results demonstrate the computational efficiency of our approach compared to a method previously proposed by the authors. 
 \end{abstract}
\vspace{-0.2cm}
\section{Introduction}
% \emph{Motivation}. 
Motivated by the need to ensure safe and smooth operation under various forms of uncertainties in many safety-critical engineering applications such as fault detection, urban transportation, attack (unknown input) mitigation and detection in cyber-physical systems and aircraft tracking \cite{liu2011robust,yong2016tcps,yong2018simultaneous}, %the worst case scenario based 
robust set-valued algorithms for state and input estimation have been recently developed to find compatible state estimates. %of states and unknown inputs. 
In addition, dynamic models of many
% Particularly, set/interval membership approaches have been broadly used to guarantee hard accuracy bounds in safety-critical bounded-error settings. Further, in 
practical systems are often only  partially known.
% , the existence of potentially dynamic unknown inputs with unknown dynamics makes the entire setting a partially unknown system. 
Thus, the development of algorithms that can combine model learning and set membership estimation approaches is 
% appropriate data-driven methods that can deal with the noisy estimated data obtained form set/interval membership approaches to estimate/approximate/abstract unknown system models is 
a critical and interesting problem.   
%adversarial settings with potentially strategic unknown inputs, it is %mostly 
%critical and desirable to simultaneously derive compatible estimates of states and unknown inputs, without assuming any \emph{a priori} known bounds/intervals for the input signals.

\emph{Literature review}.
Various model-based approaches have been proposed in the literature for designing set or interval observers across different system classes~  \cite{jaulin2002nonlinear,kieffer2004guaranteed,moisan2007near,bernard2004closed,raissi2010interval,raissi2011interval,mazenc2011interval,mazenc2013robust,wang2015interval,efimov2013interval,zheng2016design,mazenc2014interval,ellero2019unknown,yong2018simultaneous,khajenejad2019simultaneous,khajenejadasimultaneous,khajenejad2020nonlinear}, These include linear time-invariant (LTI) systems~\cite{mazenc2011interval},  linear parameter-varying (LPV) systems~\cite{wang2015interval,ellero2019unknown}, Metzler and/or partially linearizable systems~\cite{raissi2011interval,mazenc2013robust}, cooperative systems~\cite{raissi2010interval,raissi2011interval}, Lipschitz nonlinear systems~\cite{efimov2013interval}, monotone nonlinear systems~\cite{moisan2007near,bernard2004closed} and uncertain nonlinear systems~\cite{zheng2016design} systems. However, many of these methods do not account for the presence of unknown inputs—which may arise from external agents, disturbances, attacks, or simply unobserved signals—nor do they address unknown dynamics. More recent efforts have tackled the design of set-valued observers capable of simultaneously estimating both states and unknown inputs, particularly for LTI~\cite{yong2018simultaneous}, LPV~\cite{khajenejad2019simultaneous}, switched linear~\cite{khajenejadasimultaneous} and nonlinear systems~ \cite{khajenejad2020nonlinear,khajenejad2020simultaneous} considering bounded-norm noise.   
 
 On the other hand, when the system model is not precisely known, set-valued data-driven approaches have gained increasing attention in recent years~\cite{Milanese2004SetMI,canale2014nonlinear,zabinsky2003optimal,beliakov2006interpolation,calliess2014conservative}. These methods leverage input-output data to abstract or over-approximate unknown dynamics or functions, aiming to identify a set of dynamics that effectively frame or bracket the true system behavior~\cite{Milanese2004SetMI,canale2014nonlinear}. This problem has been studied under different regularity assumptions on the unknown dynamics, including univariate Lipschitz continuity~\cite{zabinsky2003optimal}, multivariate Lipschitz continuity~\cite{beliakov2006interpolation}, and H\"{o}lder continuity~\cite{calliess2014conservative}.
  
In our previous work~\cite{khajenejad2021modelstate}, we initiated an approach that combined model abstractions~\cite{Jin2020datadriven} with specific types of decomposition functions~\cite{khajenejad2020simultaneous} to construct interval framers for partially known dynamics subject to bounded noise. This framework accommodated a broad class of nonlinear state and observation mappings while treating the mapping of the unknown input dynamics as an entirely unknown function. However, the method in~\cite{khajenejad2021modelstate} required computationally expensive online optimizations to obtain affine over-approximations/abstractions of nonlinear dynamics at each time step. Moreover, it was highly case-sensitive, as it did not incorporate any stabilizing gain synthesis. As a result, the stability of the obtained interval-valued estimates was heavily dependent on the specific system dynamics. This paper proposes an alternative approach to address these limitations and ensure a more robust and systematic design. 

\emph{Contributions.}
This work bridges the gap between model-based set membership observer design approaches  and data-driven function approximation methods (i.e., model learning techniques).
Specifically, our approach extends observer design methodologies from systems with fully known yet noisy dynamics to partially unknown systems. By leveraging data-driven over-approximation and abstraction of unknown dynamics~\cite{Jin2020datadriven} we develop a recursive framework that refines the estimation of the unknown system components using noisy observation data and interval estimates. Our method constructs tight and tractable bounds for the unknown dynamics as a function of current and past interval framers, allowing these bounds to be systematically integrated into the observer system.

Furthermore, we prove that the proposed observer maintains correctness, ensuring that the framer property~\cite{mazenc2013robust} holds. We also show that our estimation and abstraction of the unknown dynamics become progressively more precise and tighter over time. Finally, we exploit Jacobian sign-stable decompositions and tight mixed-monotone decomposition of nonlinear systems~\cite{moh2022intervalACC} and formulate a tractable semi-definite program (SDP) for observer gain synthesis. This ensures input-to-state stability and guarantees the optimality of the proposed design.

\section{Preliminaries}

%\emph{Notation.} %We first summarize some notations used throughout the paper. 
{\emph{{Notation}.}} $\mathbb{N}_n$, $\mathbb{N}$, $\mathbb{R}^{n \times p}$, $\mathbb{R}^n$, and $\mathbb{R}^n_{>0}$ represent, respectively, the natural numbers up to $n$, the set of natural numbers, matrices of size $n$ by $p$, $n$-dimensional Euclidean space, and positive vectors of dimension $n$. For a vector $v \in \mathbb{R}^n$, the $p$-norm is defined as $\|v\|_{p} \triangleq \left(\sum_{i=1}^n {|v_i|^p}\right)^{\frac{1}{p}}$. For simplicity, the $2$-norm of $v$ is denoted by $\|v\|$. %and a diagonal matrix whose diagonal entries are elements of $v$ is denoted by $\diag(v)$. 
For a matrix $M \in \mathbb{R}^{n \times p}$, the $i$-th row and $j$-th column entry is denoted by $M_{ij}$. %The element-wise sign function is given by $\textstyle{\mathrm{sgn}}(M)$. 
We define $M^{\oplus} \triangleq \max(M, \mathbf{0}_{n \times p})$, {$M^{\ominus} \triangleq M^{\oplus} - M$, and $|M| \triangleq M^{\oplus} + M^{\ominus}$}, which represents the element-wise absolute value of $M$. Additionally, we use $M \succ 0$ and $M \prec 0$ (or $M \succeq 0$ and $M \preceq 0$) to indicate that $M$ is positive definite and negative definite (or positive semi-definite and negative semi-definite), respectively. All vector and matrix inequalities are  element-wise inequalities. The zero vectors in $\mathbb{R}^n$, and zero matrices of size $n \times p$ are denoted by $\mathbf{0}_{n \times p}$ and $\mathbf{0}_{n}$, respectively. Moreover, a continuous function $\alpha: [0,a) \to \mathbb{R}_{\ge 0}$ belongs to class $\mathcal{K}$ if it is strictly increasing with $\alpha(0)=0$. % (thus it is positive definite, meaning $\alpha(0)=0$ and $\alpha(x)>0$ for $x>0$). %\mk{Furthermore, a function $\alpha$ is in} class $\mathcal{K}_{\infty}$ if \mk{it is in class $\mathcal{K}$, $a = \infty$, and $\lim_{r \to \infty} \alpha(r) = \infty$; that is, $\alpha$} is unbounded. 
{Lastly, a continuous function} $\lambda : {[0,a) \times [0,\infty)} \to \mathbb{R}_{\ge 0}$ is {considered to be in} class $\mathcal{KL}$ if, for every fixed $t \geq 0$, {the function} $\lambda(s, t)$ {is in} class $\mathcal{K}$; for each fixed $s \geq 0$, $\lambda(s, t)$ {decreases with respect to $t$ and satisfies $\lim_{t \to \infty} \lambda(s, t) = 0$}.

\begin{defn}[Interval]\label{defn
} $\mathcal{I} \triangleq [\underline{z}, \overline{z}] \subset \mathbb{R}^n$ is an $n$-dimensional interval defined as the set of all vectors $z \in \mathbb{R}^{n_z}$ such that $\underline{z} \leq z \leq \overline{z}$. Interval matrices can be similarly defined. \end{defn}

\begin{prop}\cite[Lemma 1]{efimov2013interval}\label{prop:bounding}
Let $A \in \mathbb{R}^{m \times n}$ and $\underline{x} \leq x \leq \overline{x} \in \mathbb{R}^n$. Then\moh{,} $A^\oplus\underline{x}-A^{\ominus}\overline{x} \leq Ax \leq A^\oplus\overline{x}-A^{\ominus}\underline{x}$. As a corollary, if $A$ is non-negative, then $A\underline{x} \leq Ax \leq A\overline{x}$.%, where $A^+,A^{++} \in \mathbb{R}^{m \times n}$, $A^+_{i,j}=A_{i,j}$ if $A_{i,j} \geq 0$, $A^+_{i,j}=0$ if $A_{i,j} < 0$ and $A^{++}=A^+-A$.
\end{prop}
%As a straight forward corollary of Proposition \ref{prop:bounding}, 
%\begin{cor} \label{cor:bounding1}
%If $A \in \mathbb{R}^{n \times m}$ is a non-negative matrix (element-wise), then $A\underline{x} \leq Ax \leq A\overline{x}$. %, since $A^+=A$ and $A^{++}=0_{m \times n}$.
%\end{cor}
%Given $A \in \mathbb{R}^{m \times n}$, let us also define the intuitive and simplifying notation $|A| \triangleq A^++A^{++}$, with $A^+$ and $A^{++}$ given in Proposition \ref{prop:bounding}.
%\begin{lem} \label{lem:tightness}
%Suppose the assumptions in Proposition \ref{prop:bounding} hold. Then, the returning bounds for $Ax$ is tight, in the sense that $\sup\limits_{\underline{x} \leq x \leq \overline{x}}Ax=A^+\overline{x}-A^{++}\underline{x}$ and $\inf\limits_{\underline{x} \leq x \leq \overline{x}}Ax=A^+\underline{x}-A^{++}\overline{x}$, where $\sup$ and $\inf$ are considered element-wise.
%\end{lem}
%\begin{cor} \label{cor:bounding}
%\end{cor}
\begin{defn}[Lipschitz Continuity]\label{defn:lip}
The mapping $q$ is $\kappa_q$-Lipschitz continuous on $\mathbb{R}^n$, if there exists $\kappa^q >0$, such that $\|q(z_1)-q(z_2)\| \leq \kappa^q \|z_1-z_2\|$, $ \forall z_1,z_2 \in \mathcal{Z}$.
\end{defn}
%\subsection{Mixed-Monotone Systems}
%\begin{defn}[Decomposition Function]\cite[Theorem 2 and (10)-(13)]{yang2019sufficient}\label{defn:decomposition}
%\end{defn}
\begin{defn}[Jacobian Sign-Stability]
  The mapping $q$ is \emph{Jacobian sign-stable} (JSS) if the sign of
  each element of the Jacobian matrix $J_q(z)$ is constant for all
  $z \in \mathcal{Z}$.
\end{defn}
\begin{prop}[JSS Decomposition]
  {\cite[Proposition 2]{moh2022intervalACC}}
  \label{prop:JSS_decomp}
  If the Jacobian of $q$ is bounded, i.e., it satisfies
  $\underline{J}_q\le J_q(z) \le \overline{J}_q, \forall z \in \mathcal{Z}$, then $q$ can
  be decomposed as:
  \begin{align*}
  \begin{array}{r}
     q(z) = A z + \mu(z),
     \end{array}
  \end{align*}
  where the $(i,j)^{th}$ element of $A \in \real^{n \times n}$
satisfies
  \begin{align*}
  \begin{array}{r}
    A_{ij} = (\underline{J}_f)_{ij} \ \text{ or } A_{ij} = (\overline{J}_f)_{ij}
    \end{array}
  \end{align*}
  and the function $\mu$ is JSS.
%  \bulletend
\end{prop}

\begin{defn}
  \cite[Definition 4]{yang2019sufficient}
  \label{def:decomp}
  A function $q_d: \mathcal{Z} \times \mathcal{Z} \to \mathbb{R}^{n}$ is a
  \emph{mixed-monotone decomposition function} for $q$ if i)$q_d(z,z)=q(z)$, ii) $f_d$ is monotonically increasing in its first, and monotonically decreasing in its second argument.
%  \begin{enumerate}
%  \item $f_d(x,x)=f(x)$,
%  \item $f_d$ is monotonically increasing in its first argument,
%  \item $f_d$ is monotonically decreasing in its second argument.
%  \end{enumerate}
  %\bulletend
\end{defn}
 It follows directly from Definition \ref{def:decomp} that %$\forall x \in [\underline{x},\overline{x}] \subseteq \mathcal{X}$,
\begin{align}\label{eq:dec_ineq}
\begin{array}{r}
q_d(\underline{z},\overline{z}) \leq q(z) \leq q_d(\overline{z},\underline{z}) \quad \forall z \in [\underline{z},\overline{z}] \subseteq \mathcal{Z}.
\end{array}
\end{align}
\begin{prop}[Tight Mixed-Monotone Decomposition Functions]
  \cite[Proposition 4 \& Lemma 3]{moh2022intervalACC}
  \label{prop:tight_decomp}
  Suppose $q$ is JSS and has a bounded Jacobian. Then, the $i^{th}$
  component of a mixed-monotone decomposition function $q_d$ is given
  by
  \begin{align}\label{eq:tight_decomposition}
  \begin{array}{r}
    q_{d,i}(z_1,z_2) \triangleq q_i(D^i z_1 + (I_m - D^i) z_2),
    \end{array}
  \end{align}
  with
  $D^i =
  \mathrm{diag}(\max(\mathrm{sgn}((\overline{J}_q)_i),\mathbf{0}_{1\times m}))$.
  Additionally, for any interval $[\underline{z}, \overline{z}] \subseteq \mathcal{Z}$, it
  holds that $\delta^q_d \leq {F}_q e^z$, where
  ${F}_q \triangleq \overline{J}_q^\oplus + \underline{J}_q^\ominus$, $e^z \triangleq \overline{z}-\underline{z}$, and
  $\delta^f_d \triangleq f_d(\underline{z},\overline{z}) - f_d(\overline{z},
  \underline{z})$. %\bulletend
\end{prop}
\section{Problem Formulation} \label{sec:Problem}
%\vspace{-0.2cm}
\noindent\textbf{\emph{System Assumptions.}} 
Consider a partially unknown nonlinear discrete-time system with %unknown inputs and 
bounded noise %\vspace{-0.1cm}
\begin{align} \label{eq:system}
\begin{array}{ll}
x^+_{k}&=f(x_k,d_k)+Ww_k,\\
y_k&=g(x_k,d_k)+Vv_k, 
\end{array} %\\
%d_{k+1}&=h(x_k,d_k,u_k,w_k) \end{array}
\end{align}
where $x^+_k \triangleq x_{k+1}$, $x_k \in \mathcal{X} \subset \mathbb{R}^n$ is the state vector at time $k \in \mathbb{N}$, % with $\|x_0\| \leq \delta^x_0 \neq 0$, 
%$u_k \in \mathcal{U} \subset \mathbb{R}^m$ is a known input vector, 
$y_k \in \mathbb{R}^l$ is the measurement vector and $d_k \in \mathcal{D} \subset \mathbb{R}^p$ is an unknown (dynamic) 
input vector whose dynamics is governed by an \emph{unknown\footnote{Note that if the mapping $h$ is partially known (i.e., consists of the sum of a known component $\hat{h}$ and an unknown component $\tilde{h}$), we can simply consider $d_{k+1}-\hat{h}$ as the output data for the model learning procedure to learn a model of the (completely) unknown function $\tilde{h}$.} 
% augment the known component with the known input  and consider a modified vector field that is completely unknown.} 
mapping} $h:\mathbb{R}^{(n+p)} \rightarrow \mathbb{R}^{p}$:
\begin{align} \label{eq:input_dynamics}
d^+_{k}=h(x_k,d_k).
\end{align}
We refer to $z_k\triangleq\begin{bmatrix}x_k^\top & d_k^\top  \end{bmatrix}^\top \in \mathcal{Z} \subseteq \mathbb{R}^{n_z}$ as the augmented state, where $n_z \triangleq n+p$. 
The process noise $w_k \in [\underline{w},\overline{w}] \subset \mathbb{R}^{n_w}$ and the measurement noise $v_k \in [\underline{v},\overline{v}] \subset \mathbb{R}^{n_v}$ are assumed to be bounded.
%The mappings $f:\mathbb{R}^{n_z} \rightarrow \mathbb{R}^n$ and $g:\mathbb{R}^{n_z} \rightarrow \mathbb{R}^l$ are known, while %, bounded 
%the function $h=\begin{bmatrix} h^\top_1 \dots h^\top_p \end{bmatrix}^\top:\mathbb{R}^{n_z} \times \rightarrow \mathbb{R}^p$ is \emph{unknown}, but each of its arguments $h_j:\mathbb{R}^{n_z} \rightarrow \mathbb{R}$, $\forall j \in \{1\dots p\}$ is known to be Lipschitz continuous. For simplicity and without loss of generality, we assume that the Lipschitz constant $\kappa^h_j$ is known; otherwise, % with the known Lipschitz constant $\kappa^h_j$. \moha{Note that the assumption that the Lipschitz constants are known is without loss of generality, since in case that they are unknown, 
%we can estimate the Lipschitz constants with any desired precision %, applying our previously developed 
%using the approach in \cite[Equation (12) and Proposition 3]{Jin2020datadriven}.
 %Note that no assumption is made on $H$ to be either the zero matrix (no direct feedthrough), or to have full column rank when there is direct feedthrough. 
%Without loss of generality, we assume that ${\rm rk}[G^\top\; H^\top ]=p$, $n \geq l \geq 1$, $l \geq p \geq 0$ and $m \geq 0$. 
%Moreover, we assume %that 
%the following: % system assumptions hold: %, the current time variable $r$ is strictly nonnegative. % and that $x_0$ is  independent of $v_k$ and $w_k$ for all $k$. % and the unbiased 
Moreover, we assume the following.
\begin{assumption}\label{ass:known_input_output}
The bounds for disturbances and noises $ \underline{w}$, $\overline{w}$, $ \underline{v}$, and $\overline{v}$, along with the output signals $y_k$ %(output) and $u_k$ (input, if any), 
are known at every time instant. Additionally, the initial condition $x_0$ lies in $\mathcal{X}_0 = [ \underline{x}_0,\overline{x}_0] \subseteq \mathcal{X}$ with known limits $\underline{x}_0$ and $\overline{x}_0$, while the lower and upper bounds, $\underline{z}_0$ and $\overline{z}_0$, for the initial augmented state ${z}_0 \triangleq \begin{bmatrix} x^\top_0 & d^\top_0 \end{bmatrix}^\top$ {are} available, i.e., $\underline{z}_0 \leq z_0 \leq \overline{z}_0$. .
\end{assumption}%\vspace{-0.2cm}

\begin{assumption}\label{ass:mixed_monotonicity}
The mappings/functions $f:\mathbb{R}^{n_z} \rightarrow \mathbb{R}^n$ and $g:\mathbb{R}^{n_z} \rightarrow \mathbb{R}^l$ are given, locally Lipschitz and differentiable over their domains, while %, bounded 
the function $h=\begin{bmatrix} h^\top_1 \dots h^\top_p \end{bmatrix}^\top:\mathbb{R}^{n_z} \times \rightarrow \mathbb{R}^p$ is \emph{unknown}, but each of its arguments $h_j:\mathbb{R}^{n_z} \rightarrow \mathbb{R}$, $\forall j \in \{1\dots p\}$ is known to be Lipschitz continuous. For simplicity and without loss of generality, we assume that the Lipschitz constant $\kappa^h_j$ is known; otherwise, % with the known Lipschitz constant $\kappa^h_j$. \moha{Note that the assumption that the Lipschitz constants are known is without loss of generality, since in case that they are unknown, 
we can estimate the Lipschitz constants with any desired precision %, applying our previously developed 
using the approach in \cite[Equation (12) and Proposition 3]{Jin2020datadriven}. Furthermore, the lower and upper bounds of the Jacobian matrices of $f$ and $g$ are known, satisfying $\underline{J}_{s}(z) \leq J_s(z) \leq \overline{J}_{s}(z), \forall s \in \{f,g\}, \forall z \in \mathcal{Z} \triangleq \mathcal{X} \times \mathcal{D}$. %$\underline{J}^{f},\overline{J}^{f} \in \mathbb{R}^{n \times  n'}$, $\underline{J}^{g},\overline{J}^{g} \in \mathbb{R}^{n_z \times  n''}$ and $\underline{J}^{h},\overline{J}^{h} \in \mathbb{R}^{l \times  n'''}$, are known, where $n'=n+n_w, n''=n+l$, and $n'''=n+n_v$}.
\end{assumption} 
\begin{defn}[Correct Interval {Framers} and Model Resilient $\mathcal{H}_{\infty}$-Optimal Interval Observer]\label{defn:framers}
The signals/sequences $\overline{z},\underline{z}: {\mathbb{K}} \to \mathbb{R}^{n_z}$ are called upper and lower framers for the system augmented states $z_k \triangleq [x^\top_k \ d^\top_k]^\top$ of the nonlinear partially known system in \eqref{eq:system}--\eqref{eq:input_dynamics} if 
$\underline{z}_k \leq z_k \leq \overline{z}_k, \ \forall k \ge 0, {\forall w_k \in \mathcal{W},\forall v_k \in \mathcal{V}}$. 
Further, $\varepsilon_k \triangleq \overline{z}_k-\underline{z}_k$ is called the \emph{observer error} at time $k$. Any dynamical system whose states are correct framers for the augmented states of \eqref{eq:system}--\eqref{eq:input_dynamics}, i.e., with $\varepsilon_k\ge 0, \forall k\in\mathbb{K}$, is called a \emph{correct} interval framer for \eqref{eq:system}--\eqref{eq:input_dynamics}. Moreover,  the interval framer is $\mathcal{H}_{\infty}$-optimal (and if so is called interval observer), if the $\mathcal{H}_{\infty}$-gain of the framer error system $\tilde{\mathcal{G}}$ is minimized:
\begin{align}\label{eq:L1_Def}
\begin{array}{r}
\|\tilde{\mathcal{G}}\|_{\ell_2} \triangleq \sup_{\|\delta\|_{\ell_2}=1} \|\varepsilon_k\|_{\ell_2}, 
\end{array}   
\end{align}
where $\|\nu\|_{\ell_2} \triangleq \sum_0^\infty \sqrt{\|\nu_k\|^2}$ denotes the $\ell_2$ signal norm for $\nu \in \{\varepsilon,\delta\}$, respectively, while %The terms $\varepsilon_k$ and 
$\delta_k=\delta\triangleq [\delta_w^\top  \delta_v^\top]^\top$, %represent the framer error and the combined noise signals, respectively, with 
$\delta_w \triangleq\overline{w}-\underline{w}$ and $\delta_v\triangleq\overline{v}-\underline{v}$.
\end{defn}

The state and model estimation problem can be stated through a model resilient observer design, stated as follows:
\begin{problem}[Improved State \& Model Estimation]\label{prob:SISIO}
Given the partially known nonlinear system \eqref{eq:system}--\eqref{eq:input_dynamics}, design a model resilient correct $\mathcal{H}_\infty$-optimal interval observer.   such that its framer error is input-to-state stable (ISS), that is,
\begin{align}\label{eq:L1-ISS}
 \|\varepsilon_k\|_2 \leq \beta(\|\varepsilon_0\|_2,t)+\rho(\|\delta\|_{L_\infty}), \forall k \geq 0,
\end{align}
where
$\beta$ and $\rho$ are class
$\mathcal{KL}$ and $\mathcal{K}$ functions, respectively, with the $L_\infty$ signal norm given by $\|\delta\|_{L_\infty}=\sup_{t\in [0,\infty)} \|\delta_k\|_{2} =\|\delta\|_2$ and $\delta_k$ defined as in \eqref{eq:L1_Def}.
\end{problem}

\section{Model Resilient Interval Observers} \label{sec:observer}
%\subsection{Recursive Interval Observer} \label{sec:obsv}
To address Problem \ref{prob:SISIO}, we first augment \eqref{eq:system} and \eqref{eq:input_dynamics} to obtain the following augmented system:
\begin{align}\label{eq:aug_system_2}
\begin{array}{ll}
z^+_{k}&= \mu(z_k)+\hat Ww_k,\\
y_k&=g(z_k)+Vv_k, 
\end{array}
\end{align}
where $\hat W \triangleq \begin{bmatrix} W \\ \mathbf{0}_{p\times n_w} \end{bmatrix}$, $ \mu(z) \triangleq \begin{bmatrix} f(z) \\ h(z) \end{bmatrix}=\begin{bmatrix} Az+\phi(z) \\ h(z) \end{bmatrix}$, and $g(z)=Cz+\phi(z)$, while $A$ and $C$ are computed through applying Proposition \ref{prop:JSS_decomp} such that $\phi$ and $\psi$ become JSS mappings. Our goal is to design an ISS interval observer that outputs correct framers for the augmented states $\{z_k \triangleq [x^\top \ d^\top]\}_{k\ge 0}$ of \eqref{eq:aug_system_2}. The proposed observer will leverage decomposition functions (cf. Definition \ref{def:decomp} and Proposition \ref{prop:tight_decomp}) to bound the known components of the system dynamics \eqref{eq:aug_system_2}. However, the mapping $h$, as part of the augmented system, is unknown. So, we first apply our previously developed approach in~\cite{Jin2020datadriven} to obtain a data-driven over-approximation model of $h$, followed by a new result on how to tractably compute bounds for the data-driven over-approximation model as a functions of the framers to tractably include them in a to-be-designed observer structure.  
%in our setting, and to the best of our knowledge, no approach to compute decomposition functions for (partially) unknown mappings has been proposed in the literature. Hence, we have to start with computing decomposition functions for this class of functions.
\subsection{Data-Driven Abstractions \& Their Tight Bounds}
Building upon our previous result in \cite[Theorem 1]{Jin2020datadriven}, %where we developed a data-driven approach for over-approximation/abstraction of Lipschitz unknown nonlinear functions given noisy data, 
%in the model learning step %we treat 
with the history of obtained compatible intervals %in the past time steps 
up to the current time, $\{[\underline{z}_j,\overline{z}_j]\}_{j=0}^{k}$ (satisfying $\underline{z}_j \leq z_j \leq \overline{z}_j$), as the noisy input data and the compatible interval of unknown inputs, $\{[\underline{d}_j,\overline{d}_j]\}_{j=0}k$ (satisfying $\underline{d}_j \leq d_j \leq \overline{d}_j$), as the noisy output data, we can % to 
recursively construct a sequence of \emph{abstraction/over-approximation models} \{$\overline{h}_k,\underline{h}_k\}_{k=1}^{\infty}$ for the unknown input function $h$ (recall that $z \triangleq [x^\top \ d^\top]^\top$): 
\begin{subequations}
\begin{align}%\label{eq:ML}
%\nonumber &\forall j \in \{1\dots p\}: \\
& \underline{h}_{k,j}(z_k) \hspace{-0.05cm}=  \max_{\substack{t \in \{0,\dots, T-1\}}}  (\underline{d}_{k-t,j}\hspace{-0.1cm}-\hspace{-0.1cm}\kappa^{h}_j\|z_k\hspace{-0.1cm}-\hspace{-0.1cm}\tilde{z}_{k-t}\|)\hspace{-0.1cm}-\hspace{-0.1cm}\varepsilon_{k,j}, \label{lower_func}\\
& \overline{h}_{k,j}(z_k) \hspace{-0.05cm}=  \min\limits_{t \in \{0,\dots, T-1\}}  (\overline{d}_{k-t,j}\hspace{-0.1cm}+\hspace{-0.1cm}\kappa^{h}_j\|z_k\hspace{-0.1cm}-\hspace{-0.1cm}\tilde{z}_{k-t}\|)\hspace{-0.1cm}+\hspace{-0.1cm}\varepsilon_{k,j} , \label{upper_func}%[-0.2em]
 \end{align}
 \end{subequations}
where $j \in \{1\dots p\}$, $\{\tilde{z}_{k-t}=\frac{1}{2}(\overline{z}_{k-t}+\underline{z}_{k-t})\}_{t=0}^{k}$ and $\{\overline{d}_{k-t},\underline{d}_{k-t}\}_{t=0}^{k}$ are the \emph{augmented} input-output data set. Moreover, $\epsilon_{k,j}$ is the bounded abstraction/over-approximation noise added to guarantee that all possible realizations of the true function $h$ are within the provided bounds by the approximation model (cf. \cite[Theorem 1]{Jin2020datadriven} for more details).
Then, defining $\underline{h}_{k}(z_k)\triangleq [\underline{h}_{k,1}(z_k)\dots \underline{h}_{k,p}(z_k)]^\top$ and $\overline{h}_{k}(z_k)\triangleq [\overline{h}_{k,1}(z_k)\dots \overline{h}_{k,p}(z_k)]^\top$, as a result of \cite[Theorem 1]{Jin2020datadriven}, the data-driven over-approximation satisfies:
\begin{align}\label{eq:abs_ineq}
\underline{h}_{k}(z_k) \leq h(z_k) \leq \overline{h}_{k}(z_k). 
\end{align}
The following Lemma provides an approach to computing tractable bounds for the over-approximation functions $\{\overline h_k,\underline h_k\}_{k=0}^{\infty}$ using upper and lower framers.  
\begin{lem}[Tractable \& Tight Bounds for Data-Driven Abstractions]\label{lem:data_decomposition}
The data-driven over-approximation/abstraction model in \eqref{lower_func}--\eqref{upper_func} satisfies the following tight bounds:
\begin{align}\label{eq:abstraction_bound}
\begin{array}{rl}
 \underline{h}^*_{k,j}&\triangleq \max\limits_{z_k \in [\underline{z}_k,\overline z_k]}\underline{h}_{k,j}(z_k)\\  &=\hspace{-0.1cm}  \max\limits_{{t \in \{0,\dots, T-1\}}}  (\underline{d}_{k-t,j}\hspace{-0.1cm}-\hspace{-0.1cm}\kappa^h_j\|z^{*}_{k,t}\hspace{-0.1cm}-\hspace{-0.1cm}\tilde{z}_{k-t}\|)\hspace{-0.1cm}-\hspace{-0.1cm}\varepsilon_{k,j},\\
 \overline{h}^*_{k,j} &\triangleq  \min\limits_{z_k \in [\underline{z}_k,\overline z_k]}\overline{h}_{k,j}(z_k) \\
 &= \hspace{-0.1cm}  \min\limits_{t \in \{0,\dots, T-1\}}  (\overline{d}_{k-t,j}\hspace{-0.1cm}+\hspace{-0.1cm}\kappa^h_j\|z^{*}_{k,t}\hspace{-0.1cm}-\hspace{-0.1cm}\tilde{z}_{k-t}\|)\hspace{-0.1cm}+\hspace{-0.1cm}\varepsilon_{k,j},
\end{array}
\end{align}
where for $i=1,\dots,n_z$,
\begin{align}\label{eq:z*}
z^{*}_{k,t,i}=\argmax\limits_{\substack{\zeta\\ \zeta_i \in \{\underline{z}_{k,i},\overline{z}_{k,i}\},i=1,\dots,n_z}}\|\zeta-\tilde z_{k-t}\|.
\end{align}
\end{lem}
\begin{proof}
Fix $t$ and consider $l_{k,t} \triangleq \|z_k-\tilde z_{k-t}\|$ in \eqref{lower_func} and \eqref{upper_func}. Our goal is to find $z^*_{k,t} \in [\underline{z}_k,\overline{z}_k]$ that maximizes $l_{k,t}=\sqrt{\sum_{i=1}^{n_z}|z_{k,i}-\tilde{z}_{k-t,i}|^2}$. Note that $\tilde z_{k-t}$ is fixed. Moreover, for each dimension $i$ the real-valued function $|z_{k,i}-\tilde{z}_{k-t,i}|^2$ is convex in $z_{k,i}\in [\underline z_{k,i},\overline z_{k,i}]$, hence attaining its maximum in one of the corners of the real interval $[\underline z_{k,i},\overline z_{k,i}]$. So, $z^*_{k,t}$ given in \eqref{eq:z*} maximizes $l_{k,t}$ and minimizes $-l_{k,t}$.   
\end{proof}
\begin{rem}[Data-Driven Decomposition Functions]
The results in Lemma \ref{lem:data_decomposition} can be considered an implicit data-driven computation of decomposition functions. In other words, if a mapping $h$ is known, then it can bounded in an interval domain, through the mix-monotone decomposition described in Propositions \ref{prop:JSS_decomp} and \ref{prop:tight_decomp}. As a counterpart, when $h$ is unknown, but satisfies some continuity assumptions, then all possible realizations of $h$ can be bounded through the data-driven scheme in \eqref{eq:abstraction_bound}.    
\end{rem}
\subsection{$\mathcal{H}_{\infty}$-Optimal Interval Observer Design}
Given the augmented dynamics in \eqref{eq:aug_system_2} and equipped with the results in Lemma \ref{lem:data_decomposition}, we propose the following dynamic system to provide framers for the sate and unknown input signals in \eqref{eq:system} (recall that $z_k=[x_k^\top \ d_k^\top]^\top$ and similarly for their framers $\underline z_k\triangleq [\underline x_k^\top \ \underline d_k^\top]^\top, \overline z_k\triangleq[\overline x_k^\top \ \overline d_k^\top]^\top$):
\begin{align}\label{eq:observer}
%\begin{array}{rl}
\begin{cases}
 \underline x^+_k &=(A-L_1C)^\oplus \underline z_k-(A-L_1C)^\ominus \overline z_k+L_1y_k\\
 &+\phi_d(\underline z_k,\overline z_k)+L_1^{\ominus}\psi_d (\underline z_k,\overline z_k)-L_1^{\oplus}\psi_d (\overline z_k,\underline z_k)\\
 &+W^\oplus\underline w- W^\ominus \overline w+(L_1V)^\ominus \underline v-(L_1V)^\oplus \overline v,\\
 \underline d^+_k &= \underline h^*_k+(L_2C)^\ominus \underline z_k-(L_2C)^\oplus \overline{z}_k-(L_2V)^\oplus \overline v\hspace{-.1cm}+\hspace{-.1cm}L_2y_k\\
 &+L_2^\ominus \psi_d (\underline z_k,\overline z_k)-L_2^\oplus \psi_d (\overline z_k,\underline z_k)+(L_2V)^\ominus \underline v,\\ 
 \overline x^+_k &=(A-L_1C)^\oplus \overline z_k-(A-L_1C)^\ominus \underline z_k+L_1y_k\\
 &+\phi_d(\overline z_k,\underline z_k)+L_1^{\ominus}\psi_d (\overline z_k,\underline z_k)-L_1^{\oplus}\psi_d (\underline z_k,\overline z_k)\\
 &+W^\oplus\overline w- W^\ominus \underline w+(L_1V)^\ominus \overline v-(L_1V)^\oplus \underline v,\\
 \overline d^+_k &= \overline h^*_k+(L_2C)^\ominus \overline z_k-(L_2C)^\oplus \underline{z}_k-(L_2V)^\oplus \underline v\hspace{-.1cm}+\hspace{-.1cm}L_2y_k\\
 &+L_2^\ominus \psi_d (\overline z_k,\underline z_k)-L_2^\oplus \psi_d (\underline z_k,\overline z_k)+(L_2V)^\ominus \overline v,
\end{cases}
%\end{array}
\end{align} 
where $L_1$ and $L_2$ are to-be-designed observer gains, $\underline h^*_k,\overline h^*_k$ are given in Lemma \ref{lem:data_decomposition}, and $\phi_d,\psi_d$ are tight remainder-from decomposition functions of the JSS mappings $\phi,\psi$, respectively, computed through \eqref{eq:tight_decomposition} (cf. Proposition \ref{prop:tight_decomp} for more details). 
The following Lemma shows that the proposed system in \eqref{eq:observer} is indeed a correct framer for the partially known system \eqref{eq:system}--\eqref{eq:input_dynamics}, and that the abstraction/over-approximation model $\{\underline h_k,\overline h_k\}_{k\ge 0}$ becomes tighter or more precise with time.
\begin{lem}[{Correctness \& Model Improvement}]\label{lem:correctness}
Suppose the nonlinear partially known system \eqref{eq:system}--\eqref{eq:input_dynamics} (equivalently the augmented system \eqref{eq:aug_system_2}) satisfies Assumptions \ref{ass:known_input_output} and \ref{ass:mixed_monotonicity}. Then, for all $k \ge 0$ and any $w_k \in \mathcal{W}$, $v_k \in \mathcal{V}$, we have $\underline{z}_k \leq z_k \leq \overline{z}_k$. Here, $z_k$ denotes the state vector of \eqref{eq:aug_system_2} at time $k \ge 0$, and $[\underline{z}_k^\top \ \overline{z}_k^\top]^\top$ represents the state vector of \eqref{eq:observer} at the same time. Thus, the dynamical system given by \eqref{eq:observer} is a correct interval \emph{framer} for the nonlinear system \eqref{eq:aug_system_2} (equivalently for \eqref{eq:system}--\eqref{eq:input_dynamics}). Moreover, the following holds: 
\begin{align}\label{eq:tight_estimate}
\begin{array}{ll}
\underline{h}_0(z_0)\hspace{-.1cm} \leq\hspace{-.1cm} \dots\hspace{-.1cm} \leq\hspace{-.1cm} \underline{h}_k(z_k)\hspace{-.1cm} \leq \hspace{-.1cm}\dots\hspace{-.1cm} \leq \lim_{k \to \infty}\underline{h}_{k}(z_k) \hspace{-.1cm}\leq\hspace{-.1cm} {h}(z_k)\\
{h}(z_k) \hspace{-.1cm}\leq \lim_{k \to \infty}\overline{h}_{k}(z_k)\hspace{-.1cm} \leq \hspace{-.1cm}\dots\hspace{-.1cm} \leq \hspace{-.1cm}\overline{h}_k(z_k) \leq \hspace{-.1cm}\dots\hspace{-.1cm} \hspace{-.1cm}\leq \hspace{-.1cm}\overline{h}_0(z_0),
\end{array}
\end{align}
for any realization of the augmented sate sequence $\{z_k\}_{k \ge 0}$. In other words, the unknown input model estimations/abstractions are correct and become more precise or tighter with time. 
\end{lem}
\begin{proof}
Given any arbitrary $L_1,L_2$ with appropriate dimensions, the dynamics in \eqref{eq:aug_system_2} can be rewritten as 
\begin{align}\label{eq:sys_gain_include}
z^+_k=\begin{bmatrix} (A-LC)z_k+L_1(y_k-\psi(z_k)-Vv_k)+Ww_k\\
h(z_k)+L_2(y_k-Cz_k-\psi(z_k)-Vv_k).
\end{bmatrix}
\end{align}
Then, the correctness (framer) property follows from applying Propositions \ref{prop:bounding} and \ref{prop:tight_decomp} to \eqref{eq:sys_gain_include} to find upper and lower bounds for the known linear and nonlinear terms by the inequality in \eqref{eq:dec_ineq}, respectively, as well as applying Lemma \ref{lem:data_decomposition} to find $\underline h^*_k,\underline h^*_k$, i.e., bounds for the nonlinear function $h$. 

Furthermore, it directly follows from \cite[Theorem 1]{Jin2020datadriven} and the framer property (i.e., correctness) that the model estimates are correct, i.e, $\forall k \in \{0\dots \infty\}: \underline{h}_k(\zeta_k) \leq h(\zeta_k) \leq  \overline{h}_k(\zeta_k)$. Moreover, considering the data-driven abstraction procedure in \eqref{lower_func}--\eqref{upper_func}, note that by construction, the data set used at time step $k$ is a subset of the one used at time $k+1$. Hence, by \cite[Proposition 2]{Jin2020datadriven} the abstraction model satisfies \emph{monotonicity}, i.e., \eqref{eq:tight_estimate} holds.
 \end{proof}
\subsection{Input-to-State Stability}
Beyond the framer property, we provide sufficient conditions in the form of tractable semi-definite programs to design/synthesize observer gains satisfying input-to-state stability of the interval framers. This is an improvement to the approach in~\cite{khajenejad2021modelstate}, where a gain design procedure was lacking. 
\begin{thm}[$\mathcal{H}_{\infty}$-Optimal and %$L_1$ or $\mathcal{H}_{\infty}$-
ISS Observer Design]\label{thm:stability}
Suppose Assumptions \ref{ass:known_input_output} and \ref{ass:mixed_monotonicity} hold for the nonlinear partially known system in \eqref{eq:system}--\eqref{eq:input_dynamics}. Then, the following statements hold.
%\begin{enumerate}[(i)]
%\item \label{item:H_min} 
The correct interval framer proposed in \eqref{eq:observer} is $\mathcal{H}_{\infty}$-optimal with $\mathcal{H}_{\infty}$ system gain $\gamma$ if %there exists a tuple $(\gamma_*, \Delta_*,Q_*,\Omega_*,\Gamma_*,\tilde{L}_*,\tilde{N}_*,\tilde{T}_*,\tilde{M}_{x*},\tilde{Z}_*,N_*,\Phi_*)$ that solves 
the following optimization program with linear matrix inequalities (LMI) is feasible:
\begin{align}\label{eq:SDP}
\begin{array}{rl}
\hspace{-.4cm}&\min\limits_{\substack{\{\alpha,\gamma,Q,\Omega,\tilde{L}_1,\tilde{L}_2,{L}_1^p,{L}_2^p,{L}^n_{1v},{L}^p_{1v},{L}^n_{2v},{L}^p_{2v},{L}_{2c}^p,{L}_{2c}^n,T^p,T^n\}}} \gamma \\
&s.t.   \begin{bmatrix} (1-\kappa^2)Q-\alpha I & 0 & \Gamma^\top  \\
                                    0 & \gamma I & \Omega^\top  \\
                                    \Gamma & \Omega & Q 
                                     \end{bmatrix} \succ 0, \\
                                &\hspace{.5cm} Q=\text{diag}(Q_1,Q_2), \ Q_1 \in \mathbb{D}^n, \ Q_1 \in \mathbb{D}^p \ \alpha,\gamma >0,\\
 &\hspace{.5cm}\Gamma =\begin{bmatrix}T^p+T^n+Q_1F_{\phi}+(L^p_1+L^n_1)F_{\psi}\\ L^p_{2c}+L^n_{2c}+(L^p_2+L^n_2)F_{\psi}+Q_2\begin{bmatrix}\mathbf{0}_{p\times n} & I_p \end{bmatrix}\end{bmatrix},\\
 &\hspace{.5cm}\Omega = \begin{bmatrix} Q_1|W| & L^p_{1v}+L^n_{1v} & 0\\0 & L^p_{2v}+L^n_{2v} & Q_2  \end{bmatrix},\ L^p_1-L^n_1=\tilde L_1,   \\
 &\hspace{.5cm} T^p-T^n=Q_1A-\tilde L_1C, \ L^p_{2c}-L^n_{2c}=\tilde L_2C, \\
 &\hspace{.5cm} L^p_{1v}- L^n_{1v}=\tilde L_1V, \ L^p_{2v}- L^n_{2v}=\tilde L_2V, \ L^p_2-L^n_2\hspace{-.1cm}=\hspace{-.1cm}\tilde L_2,\\
 &\hspace{.5cm}
  Y^\nu \geq 0, \forall \nu \in \{p,n\},
 \forall Y\in \{T, L_{1v},L_{2v}, L_1,L_2,\tilde L_{2c}\}.
\end{array}
\end{align}
%\end{enumerate}
where 
%$\sigma = \begin{cases}1, \quad \text{if $\mathcal{G}$ is DT}, %\\ 
%0, \quad \text{if $\mathcal{G}$ is CT},\end{cases}$, $\beta = \begin{cases}1 \quad \text{if $V\ne0$}, \\ 
%0 \quad \text{if $V=0$},\end{cases}$ and
%\begin{align}\label{eq:conditions}
%\mathbf{C}=\begin{cases}
 %\Delta=\tilde W^p\!+\!\tilde W^n\!+\!(\tilde N^p\!+\!\tilde N^n)F^w_{\rho},\\
 %\tilde{T}W-\tilde{N}W_2=\tilde W^p-\tilde W^n,\ NV = N_{v}^p-N_{v}^n,\\ % \tilde N =\tilde N^p-\tilde N^n,\\
 %\Phi = \tilde{L}^p_v+\tilde{L}^n_v+\tilde{N}^p_v+\tilde{N}^n_v,\\
 %\tilde{L}_v=\tilde{L}^p_v-\tilde{L}^n_v,\  \tilde{N}_v=\tilde{N}^p_v-\tilde{N}^n_v, \\
 %\Omega = \tilde{M}^p_x\!+\!\tilde{M}^n_x+\tilde T F_{\tilde \phi}\!+\!\tilde L F_{\psi}\!+\!\tilde NF^z_{\rho},\\
 %\tilde S=\tilde S^p-\tilde S^n, \forall S \in \{T,L,N\},\\
 %\tilde{M}_z=\tilde{M}^p_z-\tilde{M}^n_z,
 %\tilde{M}_z=\tilde{T} A_z-\tilde{L}C-\tilde{N}A_2,\\
 %\tilde{T}=Q\!-\!\tilde{N}C,\
 %\gamma \!>\! 0, Q \in \mathbb{D}^n_{>0}, \tilde Y^\nu \geq 0, \forall \nu \in \{p,n\},\\
 %\forall Y\in \{W,N,L,T,L_v,N_v,M_z\},\\[-0.4cm]
%\end{cases}
%\end{align}
%while $n_{\tilde w} \triangleq n_w+p$, and $F_{\tilde{\phi}}=\begin{bmatrix} F_{\phi} & F_{\mu} \end{bmatrix}, F_{\psi}, F_{\rho}=\begin{bmatrix} F^z_{\rho} & F^w_{\rho} \end{bmatrix}$ 
$F_{\phi}$ and $F_{\psi}$ are computed by applying 
%Lemma \ref{lem:nonsmooth_decomposition} to the JSS mapping $\mu_x$, as well as 
Proposition \ref{prop:tight_decomp} to the JSS mappings $\phi$ and $\psi$, respectively, and $\kappa=\sqrt{\sum_{j=1}^p(\kappa^h_j)^2}$ is the Lipschitz constant of the mapping $h$.

Moreover, the %$L_1$ %-robust 
%or 
$\mathcal{H}_{\infty}$-robust observer gains $L_1$ and $L_2$ can be computed as $L_{i}=Q_{i*}^{-1}\tilde{L}_{i*},\ i \in \{1,2\}$, where $(\tilde{L}_{1*},\tilde{L}_{2*})$ is an optimal solution %the tuple $(Q_*,L_*,T_*,N_*)$ is the optimal solution 
to the program in \eqref{eq:SDP}.
\end{thm}
\begin{proof}
Defining $\varepsilon_k \triangleq \overline{z}_k-\underline{z}_k$, $\delta \nu \triangleq \overline \nu-\underline \nu, \forall \nu \in \{d_k,w,v\}$, $\delta^s_k \triangleq s_d(\overline{z}_k,\underline{z}_k)-s_d(\underline{z}_k,\overline{z}_k), \forall s \in \{\phi,\psi\}$, and starting from the framer system in \eqref{eq:observer}, the observer error dynamics, which by construction is a non-negative system, can be obtained:
\begin{align}\label{eq:err_sys}
\varepsilon^+_k=\begin{bmatrix}|A-L_1C|\varepsilon_k+\delta^\phi_k+L_1\delta^\psi_k+|L_1V|\delta v+|W|\delta w\\
|L_2C|\varepsilon_k+L_2\delta^\psi_k+\delta^h_k+|L_2V|\delta v \end{bmatrix},
\end{align}
where $\delta^h_k \triangleq \overline h^*_k-\underline h^*_k=[\delta^h_{k,1}\dots\delta^h_{k,p}]^\top=[(\overline h^*_{k,1}-\underline h^*_{k,1})\dots (\overline h^*_{k,p}-\underline h^*_{k,p})]^\top$.
On the other hand, from \eqref{eq:abstraction_bound} in Lemma \ref{eq:abstraction_bound}, for $j=1,\dots,p$, we derive
\begin{align}\label{eq:delta_h}
\begin{array}{rl}
\delta^h_{k,j} &\le \overline d_{k,j}-\underline d_{k,j}+2\kappa^h_j\|z^{j*}_{k,0}-\tilde{z}_k\|+2\epsilon_{k,j}\\
&=\delta^d_{k,j}+\kappa^h_j\|\delta^z_{k,j}\|+2\epsilon_{k,j},
%\min\limits_{t \in \{0,\dots, T-1\}}  (\overline{d}_{k-t,j}\hspace{-0.1cm}+\hspace{-0.1cm}\kappa^h_j\|z^*_k\hspace{-0.1cm}-\hspace{-0.1cm}\tilde{z}_{k-t}\|)\hspace{-0.1cm}-\hspace{-0.1cm}  \max\limits_{{t \in \{0,\dots, T-1\}}}  (\underline{d}_{k-t,j}\hspace{-0.1cm}-\hspace{-0.1cm}\kappa^h_j\|z^*_k\hspace{-0.1cm}-\hspace{-0.1cm}\tilde{z}_{k-t}\|) 
\end{array}
\end{align}
where the inequality holds since taking ``$\min$" and ``$-\max$" of finite number of values \eqref{eq:abstraction_bound} (corresponding to $t=0,\dots,T-1$) is less than or equal to each of the values, including the one corresponding to $t=0$. Moreover, the equality follows from the fact that $z^*_{k,0,j}=\overline z_{k,j}$ or $z^*_{k,0,j}=\underline z_{k,j}$, and $\tilde{z}_{k,j}$ is the midpoint of the interval $[\underline z_{k,j},\overline z_{k,j}]$, and hence their distance (in $2$-norm) is half of the distance between $\underline z_{k,j}$ and $\overline z_{k,j}$, i.e., $\delta^z_{k,j}$. By augmenting all the inequalities in \eqref{eq:delta_h}, together with \eqref{eq:err_sys} and by the non-negativity of the error dynamics, and given the results in Proposition \ref{prop:tight_decomp} in addition to the fact that $\delta^d_k=\begin{bmatrix}\mathbf{0}_{p\times n} & I_p\end{bmatrix} \varepsilon_k$, we obtain the following comparison system (with $\kappa^h \triangleq [\kappa^h_1\dots\kappa^h_p]^\top$ denoting the vector of Lipschitz constants): 
\begin{align}\label{eq:comparison}
\varepsilon^{+}_k \le A_z \varepsilon_k+B_z\delta \tilde w_k+f_{\ell}(\varepsilon_k). 
\end{align}
where $\tilde w_k \triangleq\begin{bmatrix}\delta w^\top & \delta v^\top & 2\epsilon^\top_k \end{bmatrix}^\top$ is the augmented noise vector, and
\begin{align*}
\begin{array}{c}
A_z \triangleq \begin{bmatrix} |A-L_1C|+F_{\phi}+|L_1|F_{\psi}\\
|L_2C|+|L_2|F_{\psi}+\begin{bmatrix}\mathbf{0}_{p\times n} & I_p \end{bmatrix}\end{bmatrix},\\
B_z=\begin{bmatrix} |W| & |L_1V| & 0 \\
0 & |L_2V| & I\end{bmatrix}.
%\tilde w_k \triangleq [\delta w^\top \delta v^\top 2\epsilon^\top_k]^\top
\end{array}
\end{align*}
Furthermore, $f_{\ell}(\varepsilon)\triangleq\|\varepsilon\|\begin{bmatrix} \mathbf{0}^\top_n & (\kappa^h)^\top\end{bmatrix}^\top$ is a Lipschitz continuous nonlinear mapping with the Lipschitz constant $\kappa=\sum_{j=1}^p(\kappa^h_j)^2$ (since $2$-norm is a Lipschitz function with its Lipschitz constant being one). By~\cite[Lemma 2]{de2002extended} the comparison system in \eqref{eq:comparison} is ISS if the following matrix inequalities are feasible:
\begin{align}\label{eq:lip_ISS}
\hspace{-.2cm}\begin{array}{rl}
&\begin{bmatrix}
(1-\kappa^2)Q-\alpha I-A_z^\top Q A_z & -A_z^\top Q B_z \\
-B_z^\top Q A_z & \gamma I-B_z^\top Q B_z
\end{bmatrix} \succ 0, \\
&Q \succ 0, \ \gamma,c >0.
\end{array}
\end{align}
By applying Schur complement, choosing a positive diagonal $Q$ (that implies $Y \in \{L,LC,LV\} \Rightarrow Q|Y|=\tilde{Y}\triangleq |QY|$), as well as the fact that we can equivalently replace $|\tilde{Y}|$ with the sum of two non-negative matrix variables ${Y}^p$ and ${Y}^n$ such that $\tilde{Y}=\tilde{Y}^p-\tilde{Y}^n$ (cf. \cite[Lemma 2 and the proof of Theorem 2]{efficientobserverecc25}), for the matrix inequalities in \eqref{eq:lip_ISS} to hold, it suffices that the program in \eqref{eq:SDP} be feasible. %This completes the proof. 
\end{proof}
%\vspace{-0.15cm}
\section{Illustrative Example} \label{sec:examples}
% We consider a slightly modified version of nonlinear dynamical system in \cite{pylorof2019design} with removing the uncertain parts of the matrices and including unknown dynamic inputs. The system can be described in the form \eqref{eq:system}--\eqref{eq:input_dynamics} with the following parameters:
We consider a slightly modified version of \moha{the continuous-time predator-prey system} in \cite{pylorof2019design}: %, with an unknown dynamic input, \moha{which} can be described as follows:
\begin{align*}
\begin{array}{rl}
     \dot{x}_{1,t} & = -x_{1,t}x_{2,t} - x_{2,t} + d_t + w_{1,t},  \\
     \dot{x}_{2,t} & = x_{1,t}x_{2,t} + x_{1,t} + w_{2,t}, \\
     \dot{d}_t & = 0.1(\cos(x_{1,t}) - \sin(x_{2,t})) + w_{d,t},
\end{array}
\end{align*}
where the (unknown input) dynamics  $\dot{d}$ is an unknown function, and the output equations are given by: \begin{align*}
\begin{array}{c}
y_{1,t}=x_{1,t} + v_{1,t},
 y_{2,t}=x_{2,t} + v_{2,t}, 
 y_{3,t} = \sin(d_t) + v_{3,t}.%1/(1+exp(d(2))-0.1*d(1)
 \end{array}
\end{align*} 
By applying the forward Euler method to discretize the system %into $x_{1, k+1} = x_{1, k} + \dot{x_1}*dt$, $x_{2, k+1} = x_{2, k} + \dot{x_2}*dt$, $d_{k+1} = d_{k} + \dot{d} *dt$ with sampling time $dt=0.1s$. 
it can be described in the form \eqref{eq:system}--\eqref{eq:input_dynamics} with the following parameters: 
$n=l=p=2$, $m=1$, $f=\begin{bmatrix} f_1 & f_2\end{bmatrix}^\top$,  
$g=\begin{bmatrix} g_1 & g_2 & g_3\end{bmatrix}^\top$, 
%$g_1(x_k)=0.2x_{1,k}+0.65x_{2,k}+0.8\sin(0.3x_{1,k}+0.2x_{2,k})$, $g_2(x_k)=\sin(x_{1,k})$, 
$u_k=0$, $w_k=[w_{1,k} \ w_{2,k} \ w_{d,k}]^\top$, $v_k=[v_{1,k} \ v_{2,k} \ v_{3,k}]^\top$, %$G=\begin{bmatrix} 0 & -0.1 \\ 0.2 & -0.2 \end{bmatrix}$, $H=\begin{bmatrix} -0.1 & 0.3 \\ 0.5 & -0.7 \end{bmatrix}$, 
$\overline{v}=-\underline{v}=\overline{w}=-\underline{w}=\begin{bmatrix} 0.1 & 0.1 & 0.1 \end{bmatrix}^\top$, $\overline{x}_0=\begin{bmatrix} 0 & 0.6 \end{bmatrix}^\top$, $\underline{x}_0=\begin{bmatrix} -0.35 & -0.1 \end{bmatrix}^\top$, where
\begin{align*}
\begin{array}{rl}
f_1(z_k)&=x_{1,k} + \delta_t(-x_{1,k}x_{2,k} - x_{2,k} + u_k + d_k + w_{1,k}),\\ 
f_2(z_k)&=x_{2,k} + \delta_t(x_{1,k}x_{2,k} + x_{1,k} + w_{2,k}),\\
% d_{k+1} 
h(z_k) &= d_k+\delta_t(0.1(\cos(x_{1,k}) - \sin(x_{2,k})) + w_{d,k})\\
g_1(z_k)&=x_{1, k} + v_{1,k},\\
 g_2(z_k)&=x_{2,k} + v_{2,k}, \
 g_3(z_k) = \sin(d_k) + v_{3,k},%1/(1+exp(d(2))-0.1*d(1)
 \end{array}
\end{align*} 
 with sampling time $\delta_t = 0.01s$. % and $\zeta^\top_k \triangleq [z^\top_k \ u^\top_k \ v^\top_k]$. 
%while the unknown input signals are depicted in Figure \ref{fig:variances2}. Note that rk$(H)=2$, thus Assumption \ref{assumption:Hfull} holds. 
Moreover, by \cite[Proposition 2]{khajenejad2021modelstate} (with abstraction slopes set to zero) we can obtain finite-valued upper and lower bounds (horizontal abstractions) for the partial derivatives of $f$ as: 
$\underline J_f=\begin{bmatrix} 0.994 & -0.01 & 0.99 \\ 0.009 & 0.9965 & -0.01 \end{bmatrix}$, $\overline J_f=\begin{bmatrix} 1.006 & -0.0065 & 1.01 \\ 0.016 & 1 & 0.01 \end{bmatrix}$.
Solving the LMIs in \eqref{eq:SDP} through Yalmip~\cite{YALMIP}, the stabilizing gains are $L_1=\begin{bmatrix}0.0543 &   0.6756  &  0.2126\\
   -0.0221 & 0.6060  & -0.0158 \\ \end{bmatrix},L_2=\begin{bmatrix} -0.2962 &   0.0911  &  0.2051 \end{bmatrix}$.
\begin{figure}[t] 
\centering
{\includegraphics[width=0.48\columnwidth]{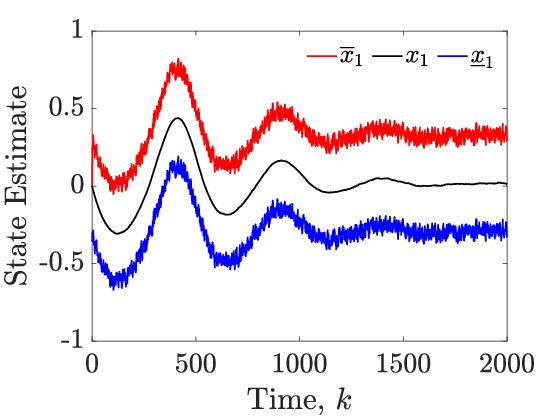}} %\label{fig:sub1} 
{\includegraphics[width=0.48\columnwidth]{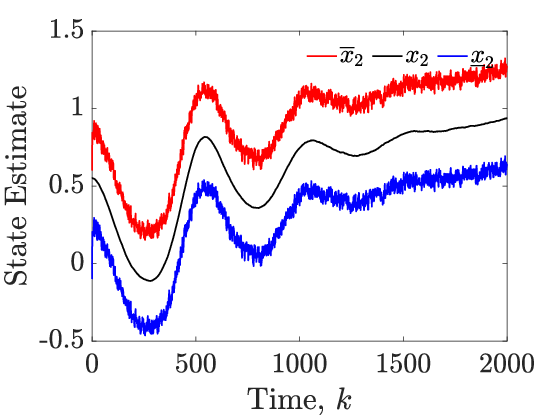}} %\label{fig:sub2}
\caption{{\small Actual states, $x_{1,k}$, $x_{2,k}$, as well as their corresponding upper and lower framers outputted by the proposed observer \eqref{eq:observer}, $\overline{x}_{1,k}$, $\underline{x}_{1,k}$, $\overline{x}_{2,k}$, $\underline{x}_{1,k}$ in a horizon of $K=2000$ time steps. \label{fig:variances2}}}% for the first 500 time steps.\label{fig:variances}}
%\end{center}
\vspace{-0.25cm}
\end{figure} 
Figure \ref{fig:variances2} shows the interval framers of the states, obtained by applying the proposed observer in \eqref{eq:observer} in a horizon of $K=2000$ time steps. For the sake of comparison, we also applied our previous approach in~\cite{khajenejad2021modelstate} which is based on computing affine over-approximation of the nonlinear dynamics at each time step, and hence, requires running  computationally expensive online optimizations. While the newly proposed approach took only $0.195287$ seconds to run and returned the framers, the one in~\cite{khajenejad2021modelstate} failed to provide results in a timely manner, i.e, we had to terminate it manually after almost ten hours. 

Furthermore, as can be seen in Figure \ref{fig:variances3}, only after we reduced the horizon to $K=250$ time steps, we observed that our previous approach in~\cite{khajenejad2021modelstate} returned estimates $\underline{x}^{\text{MZ}}_k,\overline{x}^{\text{MZ}}_k$, which were tighter than the ones outputted by the method in this paper, i.e., 
$\underline{x}_k,\overline{x}_k$. However, even for this shorter time horizon, it took $1080.79$ seconds (almost three hours) for the approach in~\cite{khajenejad2021modelstate} to return the results, while the computation time for the new approach was only $0.1201$ seconds.
\begin{figure}[t]
\begin{center}
{\includegraphics[width=0.48\columnwidth]{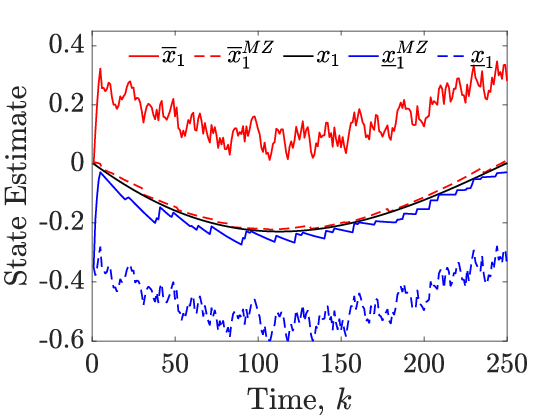}} %\label{fig:sub1} 
{\includegraphics[width=0.48\columnwidth]{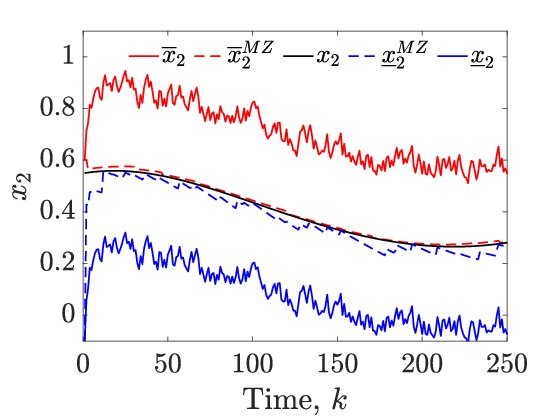}} %\label{fig:sub2}
\vspace{-0.1cm}
\caption{{\small Actual state, $x_k$,  %function, $h(x_k,d_k,u_k,w_k)$, 
as well as its estimated framers, outputted by the proposed observer \eqref{eq:observer}, $\overline{x}_{1,k}$, $\underline{x}_{1,k}$, $\overline{x}_{2,k}$, $\underline{x}_{2,k}$ in addition to the state framers obtained from the approach in~\cite{khajenejad2021modelstate}, $\overline{x}^{\text{MZ}}_{1,k}$, $\underline{x}^{\text{MZ}}_{1,k}$, $\overline{x}^{\text{MZ}}_{2,k}$, $\underline{x}^{\text{MZ}}_{2,k}$, in a horizon of $K=250$ time steps. \label{fig:variances3}}}\vspace*{-0.25cm}% for the first 500 time steps.\label{fig:variances}}
\end{center}
\end{figure} 
Finally, Figure \ref{fig:abstraction} shows the framer intervals of the learned/estimated unknown dynamics model that frame the actual unknown dynamics function $h$, i.e., satisfies \eqref{eq:abs_ineq}, as well as the global (affine) function approximation, computed via~\cite[Proposition 2]{khajenejad2021modelstate} at the initial step.
\begin{figure}[t]
\begin{center}
\includegraphics[scale=0.22,trim=15mm 3mm 10mm 10mm,clip]{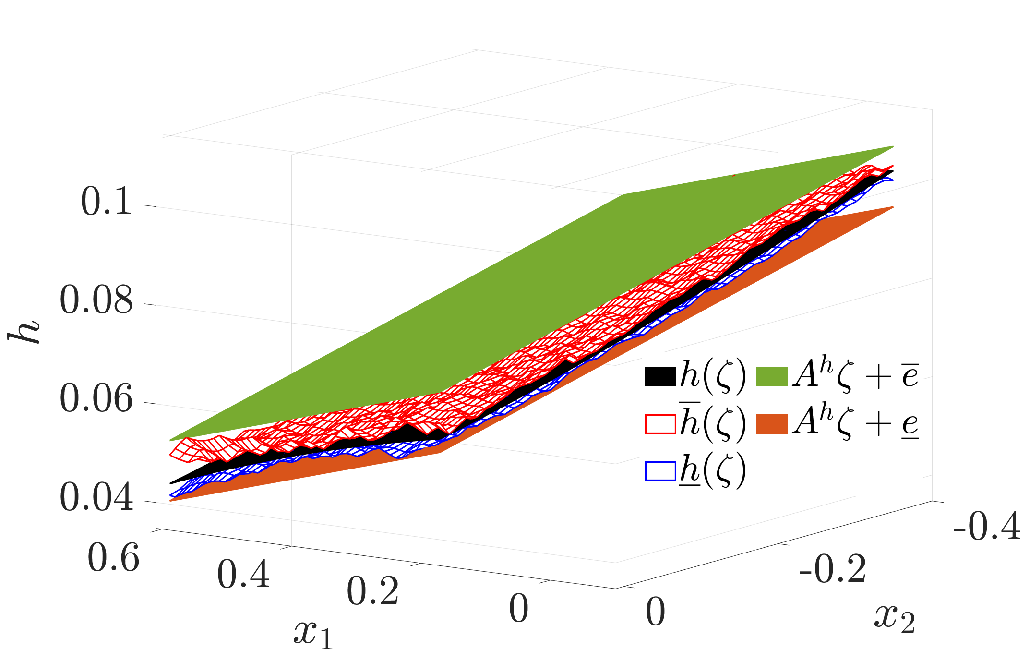}%ace{-0.15cm}
%\includegraphics[scale=0.155,trim=60m 3mm 32mm 5mm,clip]{Figures/errors-Hfull-editted8.eps}%
% \vspace{-0.05cm}
\caption{{\small Actual unknown dynamics function $h(\zeta)$, its upper and lower framer intervals %(local) abstractions 
$\overline{h}_{k},\underline{h}_{k}$ at time step $k=250$, as well as its global abstraction $A^h\zeta+\overline{e}^h,A^h\zeta+\underline{e}^h$ via \cite[Proposition 2]{khajenejad2021modelstate} at the initial step. %, at time step $k=250$.
\label{fig:abstraction}}}% for the first 500 time steps.\label{fig:variances}}
\end{center}
%\vspace{-0.15cm}
\end{figure}
%\end{discussion} 
%\vspace{-0.05cm}
\section{Conclusion \& Future Work} \label{sec:conclusion}
 In this paper, the problem of synthesizing interval observers for partially unknown nonlinear systems with bounded noise was addressed, aiming to simultaneously estimate system states and learn a model of the unknown dynamics. A framework was developed by leveraging Jacobian sign-stable (JSS) decompositions, tight decomposition functions for nonlinear systems, and a data-driven over-approximation approach, enabling the recursive computation of interval estimates that were proven to enclose the true augmented states. Tight and tractable bounds for the unknown dynamics were constructed as functions of current and past interval framers, allowing for their systematic integration into the observer design. Furthermore, semi-definite programs (SDP) were formulated to synthesize observer gains, ensuring input-to-state stability and optimality of the proposed design. Finally, simulation results demonstrated that the proposed approach outperformed the method in~\cite{khajenejad2021modelstate} in terms of computational efficiency. Future work will consider (partially) unknown continents-time and hybrid dynamical systems.
 
\bibliographystyle{unsrturl}

\bibliography{biblio}

%\vspace{-0.3cm}

\end{document}